\documentclass[9pt]{article}

\usepackage{spconf}
\usepackage[pdftex]{graphicx}

\usepackage[cmex10]{amsmath}
\usepackage{amssymb,amsbsy,amsfonts,dsfont,amsthm,bm}

\usepackage{cite}
\usepackage{algorithmicx}
\usepackage{array}
\usepackage{mdwmath}
\usepackage{mdwtab}
\usepackage[]{caption}
\usepackage[caption=false,font=footnotesize]{subfig}
\usepackage{eqparbox}
\usepackage{fixltx2e}
\usepackage{url}
\usepackage{latexsym,supertabular}

\usepackage{tikz}
\usetikzlibrary{shapes,arrows,shadows,automata,backgrounds,chains,matrix,patterns,petri,trees,calc}

\usepackage{multicol}
\usepackage{multirow}
\usepackage{booktabs}
\usepackage[capitalize]{cleveref}
\crefrangeformat{equation}{Equations~(#3#1#4) to~(#5#2#6)}

\newtheorem{theorem}{Theorem}
\newtheorem{lemma}[theorem]{Lemma}

\usepackage{etoolbox}
\newtoggle{FigPosEndPaper}
\settoggle{FigPosEndPaper}{false} 
\newbool{EqOneColumn}
\setbool{EqOneColumn}{false} 

\newcommand{\bi}{\begin{itemize}}
\newcommand{\ei}{\end{itemize}}
\newcommand{\be}{\begin{eqnarray}}
\newcommand{\ee}{\end{eqnarray}}

\newcommand{\mx}[1]{\mathbf{\bm{#1}}} 
\newcommand{\vc}[1]{\mathbf{\bm{#1}}} 
\newcommand{\Tr}[1]{\text{Tr} \left[ #1 \right]} 

\usepackage{dblfloatfix}

\DeclareMathOperator{\diag}{diag}

\renewcommand{\footnoterule}{%
	\kern -3pt
	\hrule width \columnwidth height 0.6pt
	\kern 2pt
}

\makeatletter
\renewcommand\@makefntext[1]{\leftskip=0em\hskip-0em\@makefnmark#1}
\makeatother

\title{Optimal Power Allocation for Distributed BLUE Estimation\\
	with Linear Spatial Collaboration}

\name{Mohammad~Fanaei$^{\star}$, Matthew~C.~Valenti$^{\star}$, Abbas Jamalipour$^{\dagger}$, and Natalia~A.~Schmid$^{\star}$
	\thanks{This work was supported in part by the Office of Naval Research under Award No.~N00014--09--1--1189.
	The work of M.~Fanaei is sponsored in part by the National Science Foundation under Award No.~IIA--1317103.
	}
}

\address{$^{\star}$Department of Computer Science and Electrical Engineering \\
	West Virginia University, Morgantown, WV, U.S.A. \\[5pt]
	$^{\dagger}$School of Electrical and Information Engineering \\
	University of Sydney, NSW, Australia
	}

\begin{document}
\ninept		
\maketitle

\begin{abstract}
This paper investigates the problem of linear spatial collaboration for distributed estimation in wireless sensor networks. In this context, the sensors share their local noisy (and potentially spatially correlated) observations with each other through error-free, low cost links based on a pattern defined by an adjacency matrix. Each sensor connected to a central entity, known as the fusion center (FC), forms a linear combination of the observations to which it has access and sends the resulting signal to the FC through an orthogonal fading channel. The FC combines these received signals to find the best linear unbiased estimator of the vector of unknown signals observed by individual sensors. The main novelty of this paper is the derivation of an optimal power-allocation scheme in which the coefficients used to form linear combinations of noisy observations at the sensors connected to the FC are optimized. Through this optimization, the total estimation distortion at the FC is minimized, given a constraint on the maximum cumulative transmit power in the entire network. Numerical results show that even with a moderate connectivity across the network, spatial collaboration among sensors significantly reduces the estimation distortion at the FC.
\end{abstract}


\section{Introduction}
\label{Sec:Intro}
One of the main applications of wireless sensor networks (WSNs) is {\em distributed estimation} in which spatially distributed sensors make noisy observations of (potentially correlated) signals, process their observations locally, and transmit their processed data to a central entity, known as the fusion center (FC), through communication channels corrupted by fading and additive noise. The FC will then combine the received signals to estimate either individual signals observed by local sensors or a parameter correlated with them. In most studies in the literature, it is assumed that the sensors do not communicate and/or collaborate with each other, and that the local processing is performed only on each sensor's observed noisy signal~\cite{Bahceci08,Fang09,Rashid12,Xiao06,Cui07Diversity,Fanaei2013Asilomar,Xiao08,Banavar10,Fanaei2013Milcom,Chaudhary13,Xiao05,Fanaei2012,RibeiroGiannakis06a,Ishwar2005}.
In this paper, we investigate the problem of distributed estimation under the assumption that local sensors collaborate with each other by sharing their local noisy observations. Consequently, the processing at each sensor connected to the FC will be performed on the combination of the sensor's own observations and those of the other sensors to which it has access.

Bah\c{c}eci~and~Khandani~\cite{Bahceci08} have studied a WSN in which local sensors make noisy observations of correlated Gaussian signals. They have assumed that each sensor amplifies its own local noisy observation before sending it to the FC through orthogonal channels corrupted by Rayleigh flat fading and additive white Gaussian noise. The FC will then combine the received signals from spatially distributed sensors to estimate the set of correlated signals observed by local sensors, either using the best linear unbiased estimator (BLUE) or the minimum mean squared-error estimator (MMSE). They have derived the optimal power-allocation scheme that minimizes the total cumulative transmit power in the entire network, given a constraint on the maximum estimation distortion at the FC, measured either as the estimation variance of each individual signal observed by one of the sensors, or as the average estimation variance of all signals of interest. It is crucial to emphasize that~\cite{Bahceci08} assumes that there is no communication and/or collaboration among sensors.

Kar~and~Varshney~\cite{Kar13} have studied the optimal power allocation for a WSN in which sensors collaborate with each other by sharing their local noisy observations. To the best of our knowledge, this is the first work that has considered sensor collaboration in the context of distributed estimation. In the system model studied in~\cite{Kar13}, each sensor connected to the FC, which in general could be in a subset of all sensors, forms a linear combination of its own noisy observation and the observations of other sensors to which it has access. This operation is known as the {\em linear spatial collaboration}. The sensor will then send the resulting linearly processed signal to the FC through a coherent multiple access channel~(MAC). The FC will find the linear minimum mean squared-error estimator (LMMSE) of a {\em scalar} random signal observed by spatially distributed sensors. The gains used to form the linear combinations at local sensors are optimized to minimize the LMMSE distortion, given a constraint on the maximum per-sensor or cumulative transmit power in the network. The results of their investigations show that even a moderate connectivity in the WSN could drastically reduce the estimation distortion at the FC.

As the system model of the WSN described in Section~\ref{Sec:SystemModel} shows, our goal in this paper is to generalize the network model studied in~\cite{Bahceci08} by assuming that $(a)$ the observation noises and channel noises are spatially correlated, $(b)$ a subset of sensors is {\em not} directly connected to the FC, and more importantly, $(c)$ the sensors collaborate with each other by sharing their local noisy observations through error-free, low cost links. The most important aspect of this network model is the linear spatial collaboration among local sensors. Furthermore, we will study a generalized version of the problem investigated in~\cite{Kar13}. In contrast with~\cite{Kar13}, the FC in our system model estimates the individual signals observed by distributed local sensors and {\em not} just an underlying scalar parameter that is collaboratively observed by the entire network. Moreover, we consider the communication channels between the connected sensors and the FC to be orthogonal rather than a coherent MAC. Another contribution of our work is that the FC finds the BLUE estimator of the vector of unknown signals observed by local sensors rather than the LMMSE estimator. Note that unlike the LMMSE estimator, which depends on the statistics of the signals being observed and estimated, the BLUE estimator is independent of the source statistics and is useful when the information about the signals to be estimated is limited.

As the problem formulated in Section~\ref{Sec:Analysis} and its proposed solution for the above linear spatial collaboration show, we will derive the optimal power-allocation scheme or equivalently, the optimal set of the weights used to form linear combinations of shared observations at each sensor connected to the FC. The goal of of this optimization approach is to minimize the sum of the estimation variances of the BLUE estimators for different signals observed by local sensors, given a constraint on the average cumulative transmit power in the entire network. The numerical results provided in Section~\ref{Sec:NumResults} show the applicability and effectiveness of the proposed scheme.

\section{System Model}
\label{Sec:SystemModel}
Consider a wireless sensor network (WSN) composed of $K$ spatially distributed sensors, each one of which observes a noisy version of a local signal of interest as shown in Fig.~\ref{Fig:SystemModel}. Assume that $M \leq K$ sensors are connected to a fusion center (FC). Using the received faded and noisy versions of locally processed sensor observations from a subset of sensors that are connected to it, the FC tries to find the best linear unbiased estimate (BLUE) of the vector of signals observed by individual sensors.\footnote{It is assumed that the observations of each sensor are communicated to the FC by itself if it is connected to the FC, by a subset of connected sensors to the FC with which it shares its observations if it is {\em not directly} connected to the FC, or by both.} Note that one of the major differences between the system considered here and most of the studies in the literature is that in our model (similar to~\cite{Bahceci08}), the FC estimates the individual signals observed by local sensors rather than combining the observations to estimate a set of the underlying parameters that are correlated with the collection of local observations.

\setlength{\belowcaptionskip}{-7pt}
\begin{figure}[!t]
	\centering
	\includegraphics[width=1.00\linewidth]{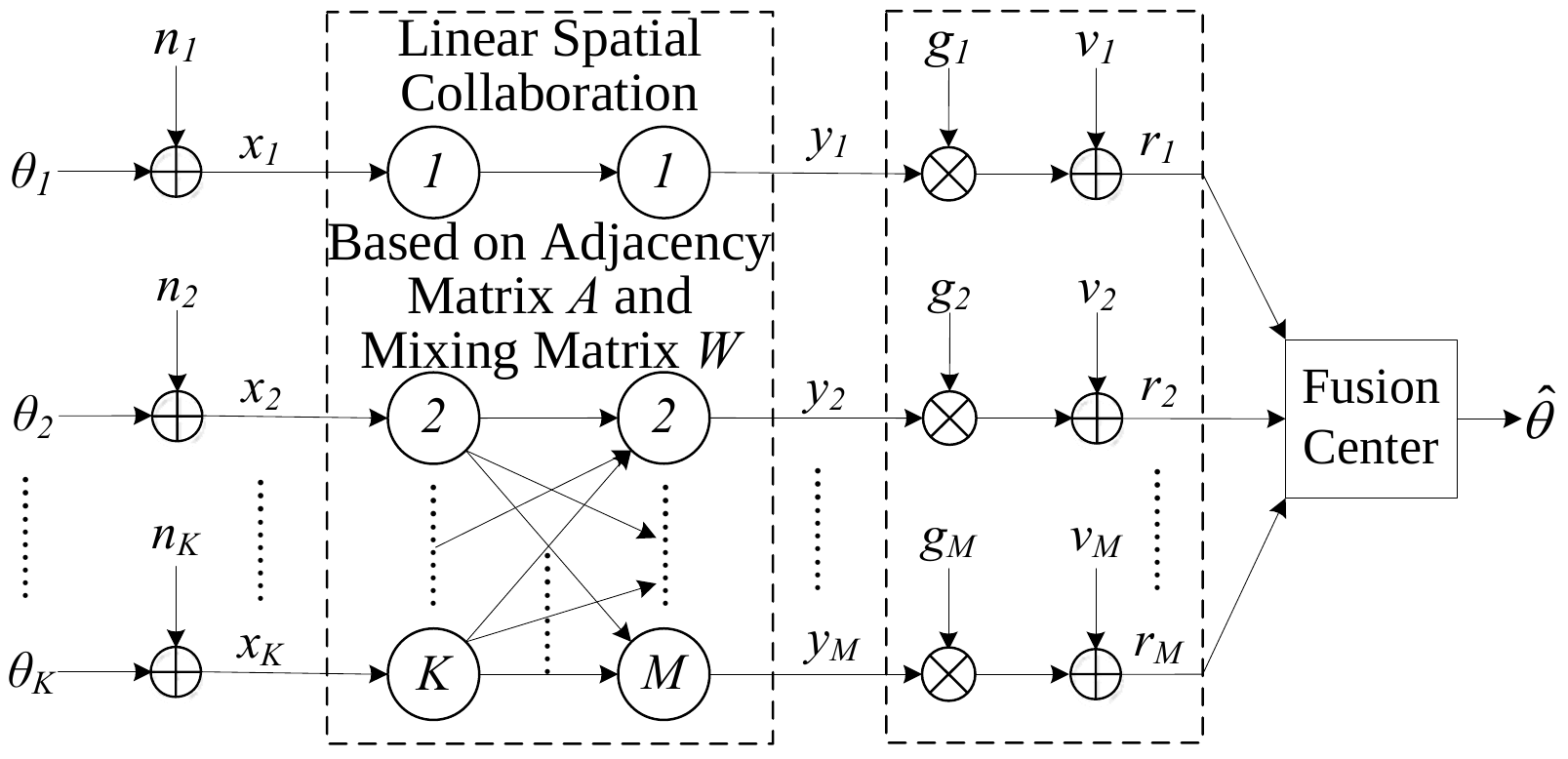}
	\caption{System model of a WSN with error-free inter-sensor collaboration in which the FC finds an estimate of $\vc{\theta} \triangleq \left[ \theta_1, \theta_2, \dotsc, \theta_K \right]^T$.
	}
	\label{Fig:SystemModel}
\end{figure}

Suppose that each sensor makes a noisy observation of a local signal of interest as
\be
x_i
&{}={}&
\theta_i + n_i,
\qquad \qquad
i = 1, 2, \dotsc, K,
\ee
where $x_i$ is the local noisy observation at the $i$th sensor, $\theta_i$ is the local unknown signal to be estimated at the FC, and $n_i$ is the observation noise. Assume that the {\em random} vector of signals observed at different sensors $\vc{\theta} \triangleq \left[ \theta_1, \theta_2, \dotsc, \theta_K \right]^T$ has zero mean and is spatially correlated with the known auto-correlation matrix $\mx{R}_\vc{\theta} \triangleq \mathbb{E} \left[ \vc{\theta} \vc{\theta}^T \right]$, where $\left( \cdot \right)^T$ represents the vector/matrix transpose operation and $\mathbb{E} \left[ \cdot \right]$ denotes the expectation operation. Furthermore, assume that the vector of observation noises $\vc{n} \triangleq \left[ n_1, n_2, \dotsc, n_K \right]^T$ is Gaussian with zero mean and $\mx{R}_\vc{n} \triangleq \mathbb{E} \left[ \vc{n} \vc{n}^T \right]$ as its auto-correlation matrix, i.e., $\vc{n} \sim \mathcal{N} \left( \vc{0}, \mx{R}_\vc{n} \right)$. It is assumed that the random vectors $\vc{\theta}$ and $\vc{n}$ are independent.

With the exception of~\cite{Kar13}, most studies in the literature assume that there is no inter-sensor communication and/or collaboration, and that each sensor processes only its own local noisy observation before transmitting it to the FC. In this paper, we assume that the sensors share their observations with each other through low cost, error-free links.\footnote{If the distance between sensors is a lot smaller than the distance between sensors and the FC, we could ignore the transmission cost of inter-sensor communications.} Suppose that the inter-sensor connectivity is modeled by an $M$-by-$K$ {\em adjacency matrix} $\mx{A}$, whose elements are either zero or one. If $\mx{A}_{j,i} = 1$, then sensor $j$ has access to the local observation of sensor $i$ through a low cost link. Otherwise, $\mx{A}_{j,i} = 0$. Note that in general, the adjacency matrix $\mx{A}$ is not necessarily symmetric since a sensor may receive the observations of a subset of other sensors, but it may not share its own observations with them. Moreover, $\mx{A}_{j,j} = 1$, $j = 1,2,\dotsc, M$, as each sensor has access to its own local observations.

Suppose that the sensors are sorted so that the first $M$ sensors are connected to the FC. Each connected sensor to the FC uses an amplify-and-forward strategy and forms a linear combination of all local observations to which it has access as
\be
y_j
&{}={}&
\sum_{\substack{
		i = 1\\
		\mx{A}_{j,i} = 1}}^K w_{j,i} x_i,
\qquad \qquad
j = 1, 2, \dotsc, M,
\ee
where $y_j$ is the transmitted signal by the $j$th sensor, and $w_{j,i}$ is the weight of the $i$th observation in the linear combination that sensor $j$ forms to be transmitted to the FC. Note that the above analog local processing can be rewritten in a vector form as $\vc{y} = \mx{W} \vc{x}$,
where $\vc{y} \triangleq \left[ y_1, y_2, \dotsc, y_M \right]^T$ is the vector of transmitted signals from the sensors that are connected to the FC, $\vc{x} \triangleq \left[ x_1, x_2, \dotsc, x_K \right]^T$ is the vector of local noisy observations, and $\mx{W}$ is an $M$-by-$K$ {\em mixing matrix}. It could easily be seen that $\mx{W}_{j,i} = 0$ if $\mx{A}_{j,i} = 0$, and $\mx{W}_{j,i} = w_{j,i}$ if $\mx{A}_{j,i} = 1$. Note that the average cumulative transmit power of the entire network can be found as
\be
P_{\mathsf{Total}}
& = &
\mathbb{E} \left[ \vc{y}^T \vc{y} \right]
\; = \;
\mathbb{E} \left[ \vc{x}^T \mx{W}^T \mx{W} \vc{x} \right] \notag \\
& = &
\Tr{ \mathbb{E} \left[ \vc{y} \vc{y}^T \right] }
\; = \;
\Tr{\mx{W} \left( \mx{R}_\vc{\theta} + \mx{R}_\vc{n} \right) \mx{W}^T},
\ee
where $\Tr{\cdot}$ denotes the trace operation of a square matrix. Therefore, the choice of the mixing matrix $\mx{W}$ affects the average cumulative transmit power of the network. Hence, determining the mixing matrix $\mx{W}$ could be considered a power-allocation strategy.

The channel between each sensor and the FC is assumed to be corrupted by fading and additive Gaussian noise. The received signal from sensor $j$ at the FC is modeled as
\be
r_j
&{}={}&
g_j y_j + v_j
,
\qquad \qquad
j = 1, 2, \dotsc, M,
\ee
where $g_j$ is the spatially independent fading coefficient of the channel between sensor $j$ and the FC, and $v_j$ is the channel noise. Note that the channels are assumed to be {\em orthogonal}. Suppose that the vector of channel noises $\vc{v} \triangleq \left[ v_1, v_2, \dotsc, v_M \right]^T$ is Gaussian with zero mean and $\mx{R}_\vc{v} \triangleq \mathbb{E} \left[ \vc{v} \vc{v}^T \right]$ as its auto-correlation matrix, i.e., $\vc{v} \sim \mathcal{N} \left( \vc{0}, \mx{R}_\vc{v} \right)$. The above model for the communication channels between local sensors and the FC could be rewritten in a vector form as
\be \label{Eq:ReceivedR}
\vc{r}
\; = \;
\mx{G} \vc{y} + \vc{v}
\; = \;
\mx{G} \mx{W} \vc{x} + \vc{v}
\; = \;
\mx{G} \mx{W} \vc{\theta} + \mx{G} \mx{W} \vc{n} + \vc{v},
\ee
where $\vc{r} \triangleq \left[ r_1, r_2, \dotsc, r_M \right]^T$ is the vector of the received signals from local sensors at the FC, and $\mx{G} = \diag \left(  g_1, g_2, \dotsc, g_M \right)$ is a diagonal $M$-by-$M$ matrix, whose $m$th diagonal element is the fading coefficient of the channel between sensor $m$ and the FC. In this paper, we assume that the FC has {\em perfect} knowledge of the {\em instantaneous} fading coefficients of the channels between local sensors and itself. This requirement could be satisfied by, for example, using pilot signals.

\section{Optimal Power Allocation for Linear Spatial Collaboration}
\label{Sec:Analysis}
It can be seen from~\eqref{Eq:ReceivedR} that, due to the independence of $\vc{n}$ and $\vc{v}$, given a realization of the vector of locally observed signals $\vc{\theta}$ and a realization of the fading coefficients of the channels between local sensors and the FC, the received vector of signals at the FC is a Gaussian random vector with mean $\vc{\mu}_{\vc{r} | \left\{ \vc{\theta}, \mx{G} \right\}} = \mx{G} \mx{W} \vc{\theta}$ and covariance matrix $\mx{R}_{\vc{r} | \left\{ \vc{\theta}, \mx{G} \right\}} = \mx{G} \mx{W} \mx{R}_\vc{n} \mx{W}^T \mx{G}^T + \mx{R}_\vc{v}$. In other words,
\be
\vc{r} \big| \left\{ \vc{\theta}, \mx{G} \right\}
\, \sim \,
\mathcal{N} \left( \mx{G} \mx{W} \vc{\theta}, \mx{G} \mx{W} \mx{R}_\vc{n} \mx{W}^T \mx{G}^T + \mx{R}_\vc{v} \right).
\ee
Upon receiving the faded and noisy version of the vector of locally processed observations, the FC finds the BLUE estimator for the vector of observed signals $\vc{\theta}$ as follows~\cite[Chapter 6]{Kay93}:
\ifbool{EqOneColumn}
{
	\be
	\widehat{\vc{\theta}}
	&{}={}&
	\left(
	\mx{W}^T \mx{G}^T
	\left( \mx{G} \mx{W} \mx{R}_\vc{n} \mx{W}^T \mx{G}^T + \mx{R}_\vc{v} \right)^{-1}
	\mx{G} \mx{W}
	\right)^{-1}
	\mx{W}^T \mx{G}^T
	\left( \mx{G} \mx{W} \mx{R}_\vc{n} \mx{W}^T \mx{G}^T + \mx{R}_\vc{v} \right)^{-1}
	\vc{r},
	\ee
}
{
	\begin{multline}
	\widehat{\vc{\theta}}
	\; = \;
	\left(
	\mx{W}^T \mx{G}^T
	\left( \mx{G} \mx{W} \mx{R}_\vc{n} \mx{W}^T \mx{G}^T + \mx{R}_\vc{v} \right)^{-1}
	\mx{G} \mx{W}
	\right)^{-1}\\
	\mx{W}^T \mx{G}^T
	\left( \mx{G} \mx{W} \mx{R}_\vc{n} \mx{W}^T \mx{G}^T + \mx{R}_\vc{v} \right)^{-1}
	\vc{r},
	\end{multline}
}
where the corresponding covariance matrix of the BLUE estimator can be found as
\be \label{Eq:EstimateCovarianceMat}
\mx{R}_{\widehat{\vc{\theta}}}
&{}={}&
\mathbb{E}
\left[
\left( \widehat{\vc{\theta}} - \vc{\theta} \right) \left( \widehat{\vc{\theta}} - \vc{\theta} \right)^T
\right] \notag \\
&{}={}&
\left(
\mx{W}^T \mx{G}^T
\left( \mx{G} \mx{W} \mx{R}_\vc{n} \mx{W}^T \mx{G}^T + \mx{R}_\vc{v} \right)^{-1}
\mx{G} \mx{W}
\right)^{-1}
\ee
Note that the calculation of the BLUE estimator is independent of the statistics of the signal to be estimated $\vc{\theta}$. As it can readily be observed from~\eqref{Eq:EstimateCovarianceMat}, the choice of the mixing matrix $\mx{W}$ affects the estimation distortion at the FC, which can be defined based on the given covariance matrix of the BLUE estimator.

The goal of this paper is to derive the optimal mixing matrix $\mx{W}$ that minimizes the total distortion in the estimation of $\vc{\theta}$ at the FC, given a constraint on the average cumulative transmit power of local sensors. We define the total estimation distortion at the FC as the trace of the covariance matrix of the BLUE estimator, which is the sum of the estimation variances for different components of $\vc{\theta}$.
This objective could be formulated as the following optimization problem:
{\setlength{\arraycolsep}{0pt}
\be \label{Eq:OptProbVer1}
\begin{aligned}
	& \underset{\mx{W}}{\text{minimize}}
	& &
	\Tr{\mx{W}^T \mx{G}^T
		\left( \mx{G} \mx{W} \mx{R}_\vc{n} \mx{W}^T \mx{G}^T + \mx{R}_\vc{v} \right)^{-1}
		\mx{G} \mx{W}
	}^{-1} \\
	& \text{subject to}
	& &
	\Tr{ \mx{W} \left( \mx{R}_\vc{\theta} + \mx{R}_\vc{n} \right) \mx{W}^T }
	\, \leq \,
	P_0
\end{aligned}
\ee
}
\\[-8pt]
\noindent where $P_0$ is the constraint on the total average transmit power in the entire network. The following lemma could be used to simplify the objective function of the above constrained optimization problem.

\begin{lemma} \label{Lemma:LowerBoundTrR}
A lower bound on { \normalfont $\Tr{\mx{R}_{\widehat{\vc{\theta}}}}$ } can be found as
{\normalfont
\be
\Tr{\mx{R}_{\widehat{\vc{\theta}}}}
\geq
\frac{K^2}{
	\Tr{\mx{W}^T \mx{G}^T
		\left( \mx{G} \mx{W} \mx{R}_\vc{n} \mx{W}^T \mx{G}^T + \mx{R}_\vc{v} \right)^{-1}
		\mx{G} \mx{W}
		}}.
\ee
}
\end{lemma}

\begin{proof}
It is proved in~\cite[Lemma 1]{Fang09} that for any arbitrary real matrix $\mx{\Phi}$ and any positive semi-definite real matrix $\mx{\Lambda}$ of proper sizes, the following inequality holds:
\be
\Tr{\mx{\Phi}^T \mx{\Lambda}^{-1} \mx{\Phi}}
& \geq &
\frac{\left(\Tr{\mx{\Phi}^T \mx{\Phi}}\right)^2}{\Tr{\mx{\Phi}^T \mx{\Lambda} \mx{\Phi}}}.
\ee
Let $\mx{\Phi} = \mx{I}_K$ and $\mx{\Lambda} = \mx{R}_{\widehat{\vc{\theta}}}$, where $\mx{I}_K$ denotes the $K$-by-$K$ identity matrix. The lower bound of the lemma would be readily derived.
\end{proof}

Using Lemma~\ref{Lemma:LowerBoundTrR}, the optimization problem of~\eqref{Eq:OptProbVer1} can be rewritten as follows:
\be \label{Eq:OptProbVer2}
\begin{aligned}
	& \underset{\mx{W}}{\text{maximize}}
	& &
	\Tr{\mx{W}^T \mx{G}^T
		\left( \mx{G} \mx{W} \mx{R}_\vc{n} \mx{W}^T \mx{G}^T + \mx{R}_\vc{v} \right)^{-1}
		\mx{G} \mx{W}
	} \\
	& \text{subject to}
	& &
	\Tr{ \mx{W} \left( \mx{R}_\vc{\theta} + \mx{R}_\vc{n} \right) \mx{W}^T }
	\, \leq \,
	P_0
\end{aligned}
\ee

\begin{lemma} \label{Lemma:ReformulateSchur}
The optimization problem given in~\eqref{Eq:OptProbVer2} is equivalent to the following form:
{\normalfont
\be \label{Eq:OptProbVer4}
\begin{aligned}
	& \underset{\mx{W},\gamma,\mx{\Gamma}}{\text{minimize}}
	& &
	\gamma \\
	& \text{subject to}
	& &
	\Tr{ \mx{W} \left( \mx{R}_\vc{\theta} + \mx{R}_\vc{n} \right) \mx{W}^T }
	\, \leq \,
	P_0 \\
	& & &
	\begin{pmatrix}
		\mx{\Gamma} & \mx{R}_\vc{n}^{-1} \\
		\mx{R}_\vc{n}^{-1} &  \mx{W}^T \mx{G}^T \mx{R}_\vc{v}^{-1} \mx{G} \mx{W} + \mx{R}_\vc{n}^{-1}
	\end{pmatrix}
	\succeq 0 \\
	& & &
	\Tr{\mx{\Gamma}} \leq \gamma
\end{aligned}
\ee
}
\\[-8pt]
\noindent where {\normalfont $\gamma$} is a real scalar, {\normalfont $\mx{\Gamma}$} is a symmetric $K$-by-$K$ real matrix, and $\mx{\Upsilon} \succeq 0$ denotes that the matrix $\mx{\Upsilon}$ is positive semi-definite.
\end{lemma}

\begin{proof}
Based on the 
Woodbury matrix inversion lemma~\cite[Page 19]{HornJohnson91}, for any arbitrary matrix $\mx{\Upsilon}$ and any non-singular matrices $\mx{\Phi}$ and $\mx{\Lambda}$ of proper sizes, if the matrix $\mx{\Phi} + \mx{\Upsilon} \mx{\Lambda} \mx{\Upsilon}^{T}$ is non-singular, then the following identity holds:
\ifbool{EqOneColumn}
{
\be
\left( \mx{\Phi} + \mx{\Upsilon} \mx{\Lambda} \mx{\Upsilon}^{T} \right)^{-1}
& = & 
\mx{\Phi}^{-1} -
\mx{\Phi}^{-1} \mx{\Upsilon}
\left( \mx{\Lambda}^{-1} + \mx{\Upsilon}^T \mx{\Phi}^{-1} \mx{\Upsilon} \right)^{-1}
\mx{\Upsilon}^T \mx{\Phi}^{-1}.
\ee
}
{
\begin{multline}
\left( \mx{\Phi} + \mx{\Upsilon} \mx{\Lambda} \mx{\Upsilon}^{T} \right)^{-1}
= 
\mx{\Phi}^{-1} - \\
\mx{\Phi}^{-1} \mx{\Upsilon}
\left( \mx{\Lambda}^{-1} + \mx{\Upsilon}^T \mx{\Phi}^{-1} \mx{\Upsilon} \right)^{-1}
\mx{\Upsilon}^T \mx{\Phi}^{-1}.
\end{multline}
}
Let $\mx{\Phi} \triangleq \mx{R}_\vc{n}^{-1}$, $\mx{\Lambda} \triangleq \mx{R}_\vc{v}^{-1}$, and $\mx{\Upsilon} \triangleq  \mx{W}^T \mx{G}^T$. Using the above matrix identity, the argument of the trace operation in the objective function of~\eqref{Eq:OptProbVer2} could be simplified as
\ifbool{EqOneColumn}
{
\be
\mx{W}^T \mx{G}^T
	\left( \mx{G} \mx{W} \mx{R}_\vc{n} \mx{W}^T \mx{G}^T + \mx{R}_\vc{v} \right)^{-1}
\mx{G} \mx{W}
&{}={}&
\mx{R}_\vc{n}^{-1}
-
\mx{R}_\vc{n}^{-1}
\left( \mx{W}^T \mx{G}^T \mx{R}_\vc{v}^{-1} \mx{G} \mx{W} + \mx{R}_\vc{n}^{-1} \right)^{-1}
\mx{R}_\vc{n}^{-1}.
\ee
}
{
\begin{multline}
\mx{W}^T \mx{G}^T
\left( \mx{G} \mx{W} \mx{R}_\vc{n} \mx{W}^T \mx{G}^T + \mx{R}_\vc{v} \right)^{-1}
\mx{G} \mx{W}
=
\mx{R}_\vc{n}^{-1}
- \\
\mx{R}_\vc{n}^{-1}
\left( \mx{W}^T \mx{G}^T \mx{R}_\vc{v}^{-1} \mx{G} \mx{W} + \mx{R}_\vc{n}^{-1} \right)^{-1}
\mx{R}_\vc{n}^{-1}.
\end{multline}
}
Hence, the optimization problem defined in~\eqref{Eq:OptProbVer2} can be rewritten as
\be \label{Eq:OptProbVer3}
\begin{aligned}
	& \underset{\mx{W}}{\text{minimize}}
	& &
	\Tr{\mx{R}_\vc{n}^{-1}
		\left( \mx{W}^T \mx{G}^T \mx{R}_\vc{v}^{-1} \mx{G} \mx{W} + \mx{R}_\vc{n}^{-1} \right)^{-1}
		\mx{R}_\vc{n}^{-1}
	} \\
	& \text{subject to}
	& &
	\Tr{ \mx{W} \left( \mx{R}_\vc{\theta} + \mx{R}_\vc{n} \right) \mx{W}^T }
	\, \leq \,
	P_0
\end{aligned}
\ee

Let $\gamma$ be a real scalar such that for any mixing matrix $\mx{W}$
\be
\Tr{\mx{R}_\vc{n}^{-1}
	\left( \mx{W}^T \mx{G}^T \mx{R}_\vc{v}^{-1} \mx{G} \mx{W} + \mx{R}_\vc{n}^{-1} \right)^{-1}
	\mx{R}_\vc{n}^{-1}
} \leq \gamma.
\ee
There exists a symmetric $K$-by-$K$ real matrix $\mx{\Gamma}$ such that~\cite{Rashid12}
\begin{subequations}
\be
\mx{R}_\vc{n}^{-1}
\left( \mx{W}^T \mx{G}^T \mx{R}_\vc{v}^{-1} \mx{G} \mx{W} + \mx{R}_\vc{n}^{-1} \right)^{-1}
\mx{R}_\vc{n}^{-1}
\preceq \mx{\Gamma}
\ee
\vspace{-0.6cm}
\be
\text{and } \; \Tr{\mx{\Gamma}} \leq \gamma
\ee
\end{subequations}
where $\mx{\Phi} \preceq \mx{\Lambda}$ means that the matrix $\mx{\Lambda} - \mx{\Phi}$ is positive semi-definite, denoted as $\mx{\Lambda} - \mx{\Phi} \succeq 0$. In other words,
\be \label{Eq:PDFirst}
\mx{\Gamma}
-
\mx{R}_\vc{n}^{-1}
\left( \mx{W}^T \mx{G}^T \mx{R}_\vc{v}^{-1} \mx{G} \mx{W} + \mx{R}_\vc{n}^{-1} \right)^{-1}
\mx{R}_\vc{n}^{-1}
\succeq 0.
\ee

Based on the Schur's complement theorem~\cite[Page~472]{HornJohnson91}, for any arbitrary matrix $\mx{\Upsilon}$ and any symmetric matrices $\mx{\Phi}$ and $\mx{\Lambda}$ of proper sizes, if $\mx{\Lambda}$ is invertible and $\mx{\Lambda} \succ 0$, then $\mx{\Phi} - \mx{\Upsilon} \mx{\Lambda}^{-1} \mx{\Upsilon}^T \succeq 0$ if and only if
\be
\begin{pmatrix}
	\mx{\Phi} & \mx{\Upsilon} \\
	\mx{\Upsilon}^T &  \mx{\Lambda}
\end{pmatrix}
\succeq 0,
\ee
where $\mx{\Lambda} \succ 0$ means that the matrix $\mx{\Lambda}$ is positive definite. Let $\mx{\Phi} \triangleq \mx{\Gamma}$, $\mx{\Upsilon} \triangleq \mx{R}_\vc{n}^{-1}$, and $\mx{\Lambda} \triangleq \mx{W}^T \mx{G}^T \mx{R}_\vc{v}^{-1} \mx{G} \mx{W} + \mx{R}_\vc{n}^{-1}$. Using the Schur's complement, the condition shown in~\eqref{Eq:PDFirst} is equivalent to the following matrix being positive semi-definite:
\be
\begin{pmatrix}
\mx{\Gamma} & \mx{R}_\vc{n}^{-1} \\
\mx{R}_\vc{n}^{-1} &  \mx{W}^T \mx{G}^T \mx{R}_\vc{v}^{-1} \mx{G} \mx{W} + \mx{R}_\vc{n}^{-1}
\end{pmatrix}
\succeq 0.
\ee

Based on the above discussions, the optimization problem given in~\eqref{Eq:OptProbVer3} is equivalent to the constrained optimization problem defined in~\eqref{Eq:OptProbVer4}, and the proof of Lemma~\ref{Lemma:ReformulateSchur} is concluded.
\end{proof}

The constrained optimization problem defined in~\eqref{Eq:OptProbVer4} is a linear programming with bi-linear matrix-inequality constraints. It could efficiently be solved using numerical solvers such as PENBMI~\cite{PenBmi}, which is fully integrated within the MATLAB$^\text{\textregistered}$ environment through version 3.0 of the YALMIP interface library~\cite{Yalmip}.

\section{Numerical Results}
\label{Sec:NumResults}
In this section, the results of numerical simulations are presented to show the effect of spatial collaboration among sensors on the estimation performance at the FC of a WSN.
Suppose that $K=6$ sensors are randomly and uniformly distributed in the two-dimensional rectangle of $\left[-10,10\right] \times \left[-5,5\right]$,
where $\times$ denotes the Cartesian product of two sets. It is assumed that all sensors are connected to the FC, i.e., $M=K$. Suppose that the covariance between the signals observed by sensors $i$ and $j$ is defined as
\be
\mx{R}_{\vc{\theta}_{i,j}}
\; \triangleq \;
\mathbb{E} \left[\theta_i \theta_j\right]
\; = \;
\sigma_\theta^2 \, \rho_{i,j},
\qquad
i,j=1,2,\dotsc,K,
\ee
where $\sigma_\theta^2$ is the variance of each component of the vector of signals to be estimated $\vc{\theta}$, and $\rho_{i,j}$ is the inter-sensor correlation coefficient that monotonically decreases with the increase of the distance between sensors as
\be
\rho_{i,j}
&{}\triangleq{}&
e^{ - \left(\frac{d_{i,j}}{\beta_1}\right)^{\beta_2} },
\qquad
i,j=1,2,\dotsc,K,
\ee
where $d_{i,j}$ is the distance between sensors $i$ and $j$, $\beta_1 > 0$ is the normalizing factor of the distances, and $0 < \beta_2 \leq 2$ controls the rate of the decay of the correlation coefficients. Note that $\rho_{i,i} = 1$, $i = 1,2,\dotsc,K$. Assume that the vector of observation (channel) noises $\vc{n}$ ($\vc{v}$) is homogeneous and equi-correlated with its covariance matrix defined as
\be
\mx{R}_{\vc{n} (\vc{v})}
&{}={}&
\sigma_{n(v)}^2 \left[ \left(1 - \lambda_{n(v)}\right) \mx{I}_{K(M)} + \lambda_{n(v)} \vc{1} \vc{1}^T \right],
\ee
where $\sigma_{n(v)}^2$ is the variance of each component of the vector of observation (channel) noises $\vc{n}$ ($\vc{v}$), $\lambda_{n(v)}$ is the constant correlation coefficients between each pair of distinct components of $\vc{n}$ ($\vc{v}$), and $\vc{1}$ is the column vector of all ones with appropriate length. In our simulations,
the communication channels between local sensors and the FC are assumed to have unit gain, i.e., $g_j=1$, $j = 1, 2, \dotsc, M$.
The following values are used for the parameters of the system to generate the simulation results presented in this section: $\sigma_\theta^2 = 1$, $\beta_1 = 6$, $\beta_2 = 3$, $\sigma_{n}^2 = 0.1$, $\sigma_{v}^2 = 0.01$, and $\lambda_{n} = \lambda_{v} = 0.1$.

Figure~\ref{Fig:Results} shows the total estimation distortion at the FC, as defined by the objective function of the optimization problem~\eqref{Eq:OptProbVer1},
versus the total average transmit power in the entire network $P_0$ for two network realizations.
Each sensor collaborates with its $q$ closest neighbors by sharing its local noisy observations with them through error-free, low cost links. Note that $q=0$ represents a network without any spatial collaboration, and $q=K-1$ corresponds to a network with full spatial collaboration.
As evident from this figure, even moderate collaboration among sensors could decrease the estimation distortion. The {\em collaboration gain} is more significant when the signals to be estimated have a higher correlation, i.e., the sensors observing them are located more closely, as depicted in Network 2.


\setlength{\belowcaptionskip}{-10pt}
\begin{figure}[!t]
	\centering
	\includegraphics[width=0.86\linewidth]{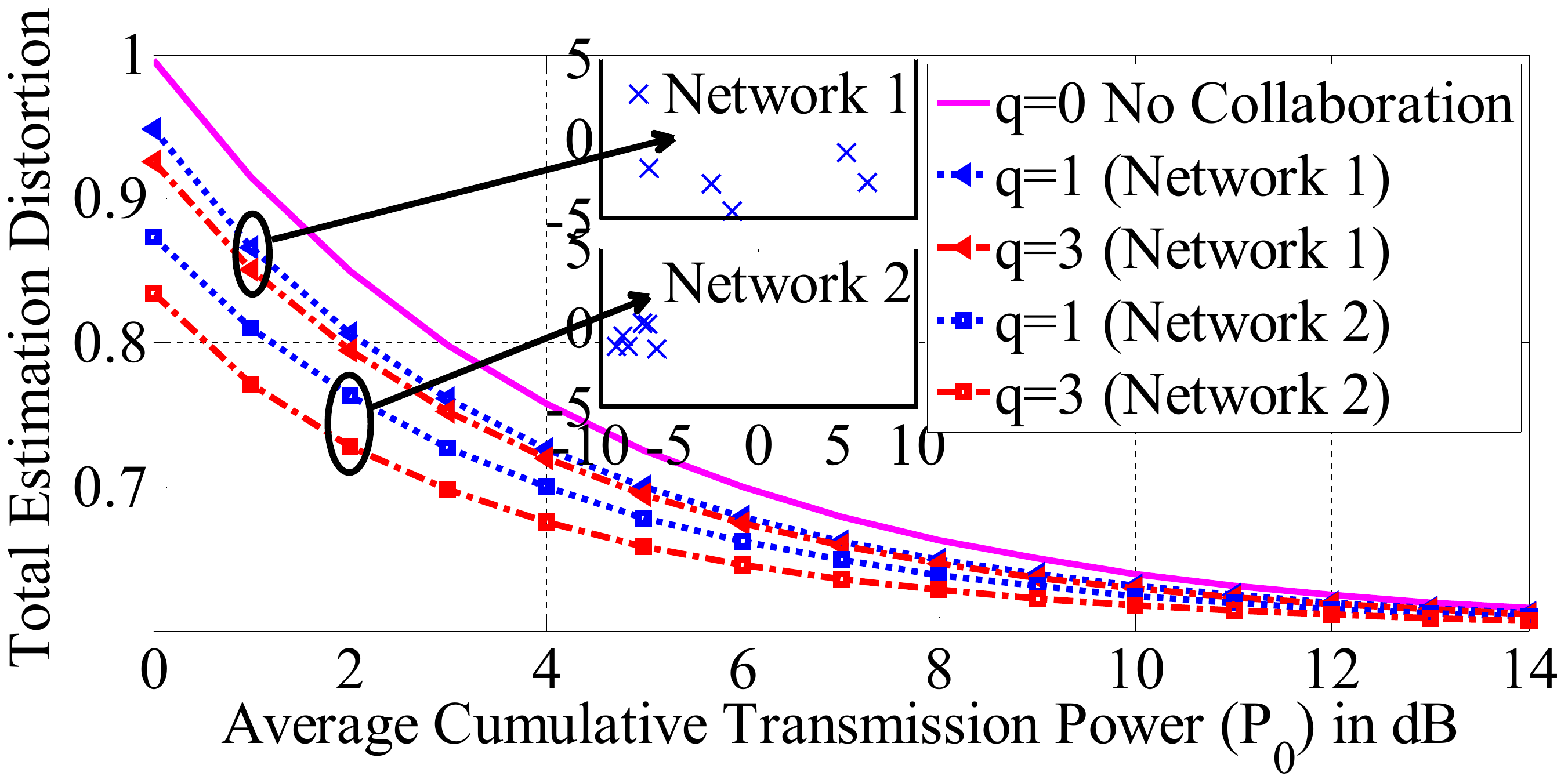}
	\vspace{-0.25cm}
	\caption{Total estimation distortion at the FC versus the average cumulative transmission power for different degrees of spatial collaboration within two random network realizations.
	}
	\label{Fig:Results}
\end{figure}

\section{Conclusions}
\label{Sec:Conclusions}
In this paper, we studied the effect of spatial collaboration on the performance of the BLUE estimator at the FC of a WSN that tries to estimate the vector of spatially correlated signals observed by sensors, rather than the well-studied case of estimating a parameter correlated with the local observations. An optimal linear spatial-collaboration scheme was derived that minimizes the sum of the estimation variances for different signals observed by the network, given a constraint on the average cumulative transmission power. The numerical results showed that even a small degree of connectivity and spatial collaboration in the network could improve the quality of the estimators at the FC.

\newpage
\bibliographystyle{IEEEtran}
\bibliography{Refs}

\begin{thebibliography}{10}
\providecommand{\url}[1]{#1}
\csname url@samestyle\endcsname
\providecommand{\newblock}{\relax}
\providecommand{\bibinfo}[2]{#2}
\providecommand{\BIBentrySTDinterwordspacing}{\spaceskip=0pt\relax}
\providecommand{\BIBentryALTinterwordstretchfactor}{4}
\providecommand{\BIBentryALTinterwordspacing}{\spaceskip=\fontdimen2\font plus
\BIBentryALTinterwordstretchfactor\fontdimen3\font minus
  \fontdimen4\font\relax}
\providecommand{\BIBforeignlanguage}[2]{{%
\expandafter\ifx\csname l@#1\endcsname\relax
\typeout{** WARNING: IEEEtran.bst: No hyphenation pattern has been}%
\typeout{** loaded for the language `#1'. Using the pattern for}%
\typeout{** the default language instead.}%
\else
\language=\csname l@#1\endcsname
\fi
#2}}
\providecommand{\BIBdecl}{\relax}
\BIBdecl

\bibitem{Bahceci08}
I.~Bah\c{c}eci and A.~Khandani, ``Linear estimation of correlated data in
  wireless sensor networks with optimum power allocation and analog
  modulation,'' \emph{IEEE Transactions on Communications}, vol.~56, no.~7, pp.
  1146--1156, July 2008.

\bibitem{Fang09}
J.~Fang and H.~Li, ``Power constrained distributed estimation with correlated
  sensor data,'' \emph{IEEE Transactions on Signal Processing}, vol.~57, no.~8,
  pp. 3292--3297, August 2009.

\bibitem{Rashid12}
U.~Rashid, H.~D. Tuan, P.~Apkarian, and H.~H. Kha, ``Globally optimized power
  allocation in multiple sensor fusion for linear and nonlinear networks,''
  \emph{IEEE Transactions on Signal Processing}, vol.~60, no.~2, pp. 903--915,
  February 2012.

\bibitem{Xiao06}
J.-J. Xiao, S.~Cui, Z.-Q. Luo, and A.~J. Goldsmith, ``Power scheduling of
  universal decentralized estimation in sensor networks,'' \emph{IEEE
  Transactions on Signal Processing}, vol.~54, no.~2, pp. 413--422, February
  2006.

\bibitem{Cui07Diversity}
S.~Cui, J.-J. Xiao, A.~J. Goldsmith, Z.-Q. Luo, and H.~V. Poor, ``Estimation
  diversity and energy efficiency in distributed sensing,'' \emph{IEEE
  Transactions on Signal Processing}, vol.~55, no.~9, pp. 4683--4695, September
  2007.

\bibitem{Fanaei2013Asilomar}
M.~{Fanaei}, M.~C. {Valenti}, and N.~A. {Schmid}, ``{Limited-feedback-based
  channel-aware power allocation for linear distributed estimation},'' in
  \emph{Proceedings of Asilomar Conference on Signals, Systems, and Computers},
  Pacific Grove, CA, November 2013.

\bibitem{Xiao08}
J.-J. Xiao, S.~Cui, Z.-Q. Luo, and A.~J. Goldsmith, ``Linear coherent
  decentralized estimation,'' \emph{IEEE Transactions on Signal Processing},
  vol.~56, no.~2, pp. 757--770, February 2008.

\bibitem{Banavar10}
M.~K. Banavar, C.~Tepedelenlio{\u{g}}lu, and A.~Spanias, ``Estimation over
  fading channels with limited feedback using distributed sensing,'' \emph{IEEE
  Transactions on Signal Processing}, vol.~58, no.~1, pp. 414--425, January
  2010.

\bibitem{Fanaei2013Milcom}
M.~{Fanaei}, M.~C. {Valenti}, and N.~A. {Schmid}, ``Power allocation for
  distributed {BLUE} estimation with full and limited feedback of {CSI},'' in
  \emph{Proceedings of Military Communications Conference (MILCOM)}, San Diego,
  CA, November 2013.

\bibitem{Chaudhary13}
M.~H. Chaudhary and L.~Vandendorpe, ``Performance of power-constrained
  estimation in hierarchical wireless sensor networks,'' \emph{IEEE
  Transactions on Signal Processing}, vol.~61, no.~3, pp. 724--739, February
  2013.

\bibitem{Xiao05}
J.-J. Xiao and Z.-Q. Luo, ``Decentralized estimation in an inhomogeneous
  sensing environment,'' \emph{IEEE Transactions on Information Theory},
  vol.~51, no.~10, pp. 3564--3575, October 2005.

\bibitem{Fanaei2012}
M.~{Fanaei}, M.~C. {Valenti}, N.~A. {Schmid}, and M.~M. {Alkhweldi},
  ``{Distributed parameter estimation in wireless sensor networks using fused
  local observations},'' in \emph{Proceedings of SPIE Wireless Sensing,
  Localization, and Processing VII}, vol. 8404, Baltimore, MD, May 2012.

\bibitem{RibeiroGiannakis06a}
A.~Ribeiro and G.~B. Giannakis, ``Bandwidth-constrained distributed estimation
  for wireless sensor networks--{P}art {I}: Gaussian case,'' \emph{IEEE
  Transactions on Signal Processing}, vol.~54, no.~3, pp. 1131--1143, March
  2006.

\bibitem{Ishwar2005}
P.~Ishwar, R.~Puri, K.~Ramchandran, and S.~Pradhan, ``On rate-constrained
  distributed estimation in unreliable sensor networks,'' \emph{IEEE Journal on
  Selected Areas in Communications}, vol.~23, no.~4, pp. 765--775, April 2005.

\bibitem{Kar13}
S.~Kar and P.~Varshney, ``Linear coherent estimation with spatial
  collaboration,'' \emph{IEEE Transactions on Information Theory}, vol.~59,
  no.~6, pp. 3532--3553, June 2013.

\bibitem{Kay93}
S.~M. Kay, \emph{Fundamentals of Statistical Signal Processing: Estimation
  Theory}, 1st~ed.\hskip 1em plus 0.5em minus 0.4em\relax NJ: Prentice Hall,
  1993.

\bibitem{HornJohnson91}
R.~A. Horn and C.~R. Johnson, \emph{Matrix Analysis}.\hskip 1em plus 0.5em
  minus 0.4em\relax Cambridge University Press, 1991.

\bibitem{PenBmi}
M.~Ko{\v{c}}vara and M.~Stingl, ``{PENBMI},'' Version 2.1, 2004, {S}ee
  \url{www.penopt.com} for a free developer version.

\bibitem{Yalmip}
J.~L{\"{o}}fberg, ``{YALMIP: A} toolbox for modeling and optimization in
  {MATLAB$^\text{\textregistered}$},'' in \emph{Proceedings of IEEE
  International Symposium on Computer Aided Control Systems Design (CACSD)},
  Taipei, Taiwan, September 2004.

\end{thebibliography}


\end{document}